\documentclass[runningheads,a4paper]{llncs}
\usepackage[margin=1.25in]{geometry}
\usepackage{amssymb}
\usepackage{graphicx}
\usepackage{verbatim}
\usepackage{mathrsfs}
\usepackage{amsfonts}

\usepackage{amsthm}
\usepackage{ifpdf}
\usepackage{amsmath}

\newtheorem{mythm}{Theorem}
\newtheorem{mylmm}{Lemma}

\newtheorem{myclm}{Claim}
\usepackage[lined,commentsnumbered,vlined]{algorithm2e}
\usepackage[noend]{algpseudocode}
\usepackage{enumerate}
\usepackage{url}
\usepackage{color,soul}
\usepackage{hyperref}
\usepackage{cite}

\urldef{\mailsa}\path|{ibanerje, richards}@cs.gmu.edu|    
\newcommand{\keywords}[1]{\par\addvspace\baselineskip
\noindent\keywordname\enspace\ignorespaces#1}

\begin{document}

\mainmatter  % start of an individual contribution

% first the title is needed
\title{Sorting Under 1-$\infty$ Cost Model}

% a short form should be given in case it is too long for the running head
\titlerunning{Sorting Under 1-$\infty$ Cost Model}

% the name(s) of the author(s) follow(s) next
%
% NB: Chinese authors should write their first names(s) in front of
% their surnames. This ensures that the names appear correctly in
% the running heads and the author index.
%
\author{Indranil Banerjee, Dana Richards}
\authorrunning{I  Banerjee, D Richards}
% (feature abused for this document to repeat the title also on left hand pages)

% the affiliations are given next; don't give your e-mail address
% unless you accept that it will be published
\institute{George Mason University\\ Department Of Computer Science\\ Fairfax Virginia 22030, USA\\
\mailsa}

%
% NB: a more complex sample for affiliations and the mapping to the
% corresponding authors can be found in the file "llncs.dem"
% (search for the string "\mainmatter" where a contribution starts).
% "llncs.dem" accompanies the document class "llncs.cls".
%

\toctitle{Lecture Notes in Computer Science}
\tocauthor{Authors' Instructions}
\maketitle

\begin{abstract}
In this paper we study the problem of sorting under non-uniform comparison costs, where costs are either 1 or $\infty$. If comparing a pair has an associated cost of $\infty$ then we say that such a pair cannot be compared (forbidden pairs). Along with the set of elements $V$ the input to our problem is a graph $G(V, E)$, whose edges represents the pairs that we can compare incurring an unit of cost. Given a graph with $n$ vertices and $q$ forbidden edges we propose the first non-trivial deterministic algorithm which makes $O((q + n)\log{n})$ comparisons with a total complexity of $O(n^2 + q^{\omega/2})$, where $\omega$ is the exponent in the complexity of matrix multiplication. We also propose a simple randomized algorithm for the problem which makes $\widetilde{O}(n^2/\sqrt{q + n} + n\sqrt{q})$ probes with high probability. When the input graph is random we show that $\widetilde{O}(\min{(n^{3/2}, pn^2)})$ probes suffice, where $p$ is the edge probability. 
\end{abstract}

\section{Introduction}
Comparison based sorting algorithms is one of the most studied area in theoretical computer science. However, the majority of the efforts have been focused on the uniform comparison cost model. Arbitrary non-uniform cost models can make trivial problems non-trivial, like finding the minimum \cite{4,8}. Thus it makes sense to consider a more structured cost. For example, a common cost model is the monotone\footnote{By monotone we mean that the cost of comparing a pair is a monotone function of the values of the pair.} cost model. As shown in \cite{8} the best one can do is to get an algorithm that is within a logarithmic factor of a cost optimal algorithm. However, the 1-$\infty$ cost model in this paper is not monotonic. This model has comparison cost of 1 or $\infty$. A pair with cost $\infty$ is considered a ``forbidden pair''. The set of pairs with comparison cost 1, defines an undirected graph, $G(V, E)$, where $V$ is the set of keys and $E$ represents the allowed comparisons. We call $G$ the comparison graph. Define $E_f$ to be the set of forbidden pairs. Let $|V| = n$ and $|E_f| = q$. 

An example of a problem that uses this model is the nuts and bolts problem. This is strictly not a sorting problem rather a matching one. In this problem one is given two sets of elements, a set of nuts and a set bolts. Elements in each set have distinct sizes and for each nut it is guaranteed that there exists a unique bolt of same size. Matching is performed by comparing a nut with a bolt. However, pairs of nuts or pairs of bolts cannot be compared. So in this case $G = K(N, B)$ is a complete bipartite graph with edges from the set of nuts $N$ to the set of bolts $B$. This problem has been solved in the mid  1990s \cite{7,9}. The existence of a $O(n \log{n})$ time deterministic algorithm was proved for it using the theory on bipartite expanders \cite{7}. 

The problem of sorting with forbidden pairs is still open for the most part. It is closely related to the problem of partial sorting under a relation determining oracle. In this model we are given a set $P$ of elements and a oracle $\mathcal{O}_r$ which is used to determine the relations between pairs of elements in $P$. The goal is to determine all the valid relations. Number of queries made to $\mathcal{O}_r$ is defined as the {\it query complexity}. Since there are $\Omega(2^{n^2/4})$\cite{16} labelled posets with $n$ elements, it immediately follows that the information theoretic bound (ITB) for the query complexity is $\Omega(n^2)$. This is has been investigated for width bounded posets in \cite{17}, where the authors show that if $P$ has width at most $w$ then the ITB for the query complexity is $O((w + \log{n})n)$. They presented a query optimal algorithm for width bounded posets whose total complexity is $O(nw^2\log{n \over w})$. Their main contributions were on developing an efficient data structure which was use to store a poset as disjoint chains and queries unknown relations using a weighted binary search method. This algorithm can be generalized for any poset with an additional $\log{w}$ factor added to the the query complexity. Their results were the first major extension in this line of research after the seminal work by Faigle and Tur$\acute{\mbox{a}}$n\cite{18} which showed existence of such an algorithm. Although an efficient implementation of it were not known at the time. Another similar problem is the {\it local sorting} problem.
In this problem $V$ is an ordered set and for each $(u, v) \in E$ we want to determine their relative order. The problem is to determine if this can be done without resorting to sorting the entire set $V$, since the ITB for this problem is $\Omega(n \log {\Delta})$ in the standard comparison tree model (where $\Delta$ is the maximum degree of $G$).  Currently no non-trivial deterministic algorithm is known for this problem. However, there is a randomized algorithm which makes optimal number of comparison with high probability \cite{20}.

The query model used in this paper differs from \cite{17} in  following manner: we don't charge for checking whether an edge exists but we only charge for the comparisons made. The number of comparisons made or rather asked to the oracle is naturally defined as the {\it comparison complexity} or the {\it probe complexity}. However, no non-trivial ITB for the probe complexity is known in the standard decision tree model. We believe that the model is too weak for this purpose. For example, given a comparison graph $G$ the number of different acyclic orientations of $G$ gives an upper bound on the number of possible answers. Given the fact that it does not take any comparisons to identify $G$ (up to isomorphism) and $G$ has at most $\le \prod_{v \in V}(d_v + 1) \le n^n$\cite{20} number of acyclic orientations the ITB of $O(n \log{n})$ in the standard comparison tree model is too week for this problem. The matter is further complicated if one is also given the guarantee that the graph $G$ is {\it sortable}. We say $G$ is sortable $G$ can be totally sorted. This restriction further reduces the number of possible answers for graphs with small number of edges. For example if $G$ has $\le n-1$ edges then we can determine the unique total order by just making one comparison. Since any acyclic orientation of the edges of $G$ must give a hamiltonian path and $G$ has $\le n-1$ edges, the edges must link consecutive vertices in the unknown order. A solitary probe is then used to determine the direction of this ordering. In this paper we take $G$ to be arbitrary and not necessarily sortable. Hence by sorting $G$ we mean determining the orientations of the edges of $G$ such that the resulting partial order (which is unique) has the maximum number relations.

In this paper we propose the first non-trivial deterministic algorithm under the probe complexity model as well as a randomized algorithm. The results are expressed in terms of $n$ and $q$. Expressing the results in terms the number of forbidden edges fits naturally with the problem. First of all $q$ and $w$ are related. Let $P_G$ be the poset found after sorting $G$. We have $q \ge \mbox{\# of incomparable pairs in}\ P_G \ge {w \choose 2}$. Hence, $w = O(\sqrt{q})$. Although we cannot directly compare the probe complexity  used in this paper with the query complexity in \cite{16} but it gives a better sense of the relatedness of the two models. Secondly, in the absence of any other structural properties of the input graph $G$, $q$ gives a good indication of how difficult it is to sort $G$. For example, when $q = O(\log{n})$, it is easy to see that one can sort in $O(n \log{n})$ total time. To do this we pick an arbitrary pair of non-adjacent vertices and take out one of them, removing it from the graph. We do the same thing with the remaining graph until the graph remaining is a clique. It is clear that we had to take out at most $O(\log{n})$ vertices. Then we sort this graph with $O(n \log(n))$ comparisons and merge the vertices we had remove previously by probing all the remaining undirected edges, which is at most $O(n \log{n})$. On the other extreme, if $|E| = {n \choose 2} - q = O(n)$ then it can be shown that we need to make $\Omega(|E|)$ probes to determine the partial order, since the complete bipartite graph $K(A,B)$ with $|A| \ll |B|$ has many acyclic orientations\cite{19,20}. So in this case one has to probe most of the allowed edges.

In the context of randomized algorithms, this problem has been studied in \cite{1, 15}. The authors in \cite{1} proposed a randomized algorithm that sorts $G$ with a probe complexity of $\widetilde{O}(n^{3/2})$ with high probability\footnote{By \textit{high probability} we mean that the probability tends to 1 as $n \to \infty$.}. However their implementation uses as a sub-routine a poly-time uniform sampling algorithm to sample points from a convex polytope\cite{21}. The authors did not discuss the exact bound on the total complexity in their paper. At each step the algorithm either finds a {\it balancing edge}\footnote{An edge in $G$ revealing whose orientation is guaranteed to reduce the number of linear extension of the current partial order by a constant fraction. The pair of vertices incident to this edge is referred to as a {\it balancing pair}.} or finds a subset of elements that can be sorted quickly. For an arbitrary $G$ it is not guaranteed that a {\it balancing edge} always exists. However, when $G$ is the complete graph there always exists a {\it balancing edge} that reduces the number of linear extension at-least by a factor of $8/11$ \cite{13}.

\subsection{Our Results}
The main contributions of this paper are as follows:

\begin{itemize}
   \item Given a comparison graph $G$ we propose a deterministic algorithm that sort $G$ with $O((q + n)\log{n})$ probes. The total complexity of our algorithm is $O(n^2 + q^{\omega/2})$, where $\omega \in [2, 2.38]$ is the exponent in the complexity of matrix multiplication. We use only elementary methods in our algorithm. We start by finding a set of large enough cliques in $G$ and use its elements to determine a good pivot. This is then applied recursively to induced subgraphs of $G$ to generate a collection of partial orders. We then merge these partial orders in the final step. 
   
   \item We propose a randomized algorithm which sorts $G$ with $O(n^2/\sqrt{n + q} + n\sqrt{q})$ probes with high probability. We use a random graph model for this purpose. The method uses only elementary techniques and unlike in \cite{1} has a total run time of $O(n^\omega)$ in the worst case.
   
   \item When $G$ is a random graph with edge probability $p$ we show that one can sort $G$ with high probability using only $\widetilde{O}(\min{(n^{3/2}, pn^2)})$ probes. 
\end{itemize}

\noindent The rest of this paper is organized as follows: in section 1.2 we introduce some definitions and lemmas for later use. Section 2 details the proposed deterministic algorithm. In section 3 we introduce the randomized algorithm and its extension to random graphs.
%------------------------------------------------------------------------------
%
%
%
%------------------------------------------------------------------------------
\subsection{Definitions}
Recall $G(V, E)$ is the input graph on the set $V$ of elements to be sorted. A pair of vertices $(u, v)$ can be compared if $(u, v) \in E$, otherwise, we say the pair is forbidden and is in $E_f$. The graph $G$ is given to us by our adversary. Let $G_i$ be the graph after $i$-edges have been oriented and $P_i$ be the associated partial order. We denote the degree of a vertex $v$ by $d(v)$ and $n(v) = n - 1 - d(v)$ is the number of vertices that are not adjacent to $v$. The set of neighbors of a vertex $v$ is denoted by $N(v)$. We use the notation $E(A,B)$ we denote the set of edges between the sets of vertices $A, B \subset V$. We also define the little-$o$ notation to remove any ambiguity from our exposition.

\begin{definition}
   If $f(n) \in o(g(n))$ then $f(n) \in O(g(n))$ but $f(n) \not \in \Omega(g(n))$.
\end{definition}

\begin{mylmm}
	Let $\{f_1(n),f_2(n),...,f_k(n) \}$ be a finite set of non-negative monotonically increasing functions in $ n $ such that:
	\begin{enumerate}
		\item $ \forall i\ f_i(n) \in o(g(n)) $ 
		\item $ \sum_{i}{f_i(n)} \le cg(n)$
	\end{enumerate}
	
\noindent If $ F(n) = \sum_{i}{f_i^2(n)}$ then $F(n) \in o(g^2(n))$.
\end{mylmm}

\begin{proof}
   
   See appendix.
\end{proof}

\begin{mylmm}
   Let $T(n) = \sum_{i=1}^{k}{T(n_i)} + f(n)$ where $\sum_i{n_i} \le \delta n$ for some $0 < \delta < 1$ and $f(n) \in o(n^2)$. Then, $T(n) \in o(n^2)$.
\end{mylmm}
\begin{proof}
See appendix.
\end{proof}

\section{A Deterministic Algorithm For Restricted Sorting}
First we look at a simple case where $q = O(n)$. We will use some of the main ideas from this algorithm to extend it to the general case. This initial algorithm will have a worse probe complexity than the main algorithm.

\subsection{A Restricted Case}
Assume $q \le cn$ for some constant $c$. Let $R = \{v \in V\ |\ n(v) > c_1\}$ for some constant $c_1$. Then $|R| \le  (2c/c_1)n$. We choose $c_1 = 4c$. This is obvious from the fact that $\sum_{v}{n(v)} \le 2cn$. Let $S = V \setminus R$ and $G[S]$ be the induced subgraph generated by $S$. We have $|S| \ge n/2$ and  if $v \in S$ then $n(v) \le c_1$. 

\begin{myclm}
   There exists a subset $X \subset S$ such that $|X| \ge n/2(4c+1)$ and $G[X]$ is a complete graph.
\end{myclm}

\begin{proof}
	Let us construct $X$ explicitly. We start with $X = {u}$, where $u$ is an arbitrary vertex in S. We pick successive vertices from $S$ iteratively. Let $v$ be last vertex to be added to $X$. Since $ v $ has at least $n-c_1$ neighbors, whenever we pick a neighbor of $v$ from $S$ to add to $X$ we loose at most $c_1 + 1$ vertices (which include the vertex we picked). Hence if we pick neighbors of $v$ the size of $X$ is at least $|S|/ (c_1 + 1) \ge n/2(4c+1)$.
	
\end{proof}

\noindent Clearly the above procedure runs in $O(n^2)$ time and makes no comparisons. Now we are ready to describe our algorithm. The main algorithm is recursive and we have two levels of recursion. We shall break the algorithm into several steps.

\subsubsection{Initial Sorting:}
Given the input graph $G$, let $X$ be a clique, with $|X| \ge n/2(4c+1)$ (Claim 1). Let $Y = V \setminus X$. Note that $|Y| \le n - n/2(4c+1) = (8c+1/8c+2)n$. Now we sort $X$ using $O(n\log{n})$ comparisons as $G[X]$ is a complete graph. We can use a standard comparison based sorting algorithms for this purpose. Now we have two possibilities: 
\begin{itemize}
\item[Case 1:] If $|Y| = o(n)$\footnote{Note that ``$|Y| = o(n)$'' is not an algorithmic test. We use it in this algorithm to establish a framework for the second algorithm, which uses a traditional test.}, then we probe all edges of $G[Y]$ and $G[Y, X]$, where $G[Y, X]$ is the induced bipartite graph generated by the sets $Y$ and $X$. Then we take the transitive closure of the resulting relations, which does not need any additional probes. It can be easily seen that the number of probe made in the previous step is $o(n^2)$. For the sake of contradiction if we assume that it is not so then $|X||Y| + |Y|^2/2 \ge dn^2$ for some $d$. Which implies $|Y| \ge dn$, since $|X| + |Y|/2 \le n$. But then, $|Y| = \Omega(n)$, which is not true according to our earlier assumption. So, in this case we would have sorted $V$ by making only $o(n^2)$ probes.
      \item[Case 2:] Otherwise $|Y| \ge \delta n$, for some constant $\delta$. In this case we recursively partition $Y$ based on elements from $X$. We call this the partition step.

\end{itemize}

\subsubsection{Partition step:}
We will recursively partition both $X$ and $Y$. To keep track of the current partition depth we rename $X$ to $X_{00}$ and $Y$ to $Y_{00}$. We pick $m_{00}$ the median of $X_{00}$ (after $X_{00}$ is sorted). Since $X_{00} \subset S$ we have $n(m_{00}) \le c_1$. So $m_{00}$ will be comparable to all but at most $c_1$ elements of $Y_{00}$. Let,\[A_{00} = \{v \in Y_{00}|\ v \in N(m_{00})\}\] and $B_{00} = Y_{00} \setminus A_{00}$. Note $|B_{00}| \le c_1$. Now let $U_{00}$ be the subset of $A_{00}$ whose elements are $\ge m_{00}$ and the set $L_{00}$ accounts for the rest of $A_{00} \setminus m_{00}$. Let $X_{10}$ and $X_{11}$ be the elements of $X_{00}$ that are $<$ and $\ge$ to $m_{00}$ respectively. We recursively partition the sets $U_{00}$ and $L_{00}$ using the medians of $X_{10}$ and $X_{11}$.The $B$-sets are kept for later processing. We rename the sets  $U_{00}$ and $L_{00}$ to $Y_{10}$ and $Y_{11}$. So, the pairs $(X_{10}, Y_{10})$ and $(X_{11}, Y_{11})$ are processed as above generating the sets $A_{10}, A_{11}, B_{10}$ and $B_{11}$. We continue doing this until the size of the $X$-set is $\le c_2$, where $c_2$ is some constant. At this point we don't know the size of the $Y$-set paired with it. There are two cases we need to consider:

\begin{enumerate}
   \item[Case 1:] $|Y| = o(n)$ Then we probe all the edges of $G[Y]$ and $G[X, Y]$ which uses at most $c_2|Y| + {{|Y|} \choose {2}}$ number of comparisons. 
   
   \item[Case 2:] $|Y| \ge \delta n$. Then we have $|Y| \ge \delta n$ for constant $\delta$. Hence the graph $G[Y]$ can have at most $\le {(c / {\delta})} |Y|$ missing edges. This satisfies our initial premise that the number of missing edges in $G[Y]$ is linear in the number of vertices. Hence we can apply our initial strategy recursively. That is we first find a large enough clique (which according to Claim 3 must exist) and then use it to partition the rest of the set $Y$.  
 
\end{enumerate}
Let us visualize using a partial recursion tree $T$ (see Fig.1 below). We shall call $T$ the partial recursion tree for reasons soon to be clear. At the root we have the pair $(X_{00}, Y_{00})$. It has two children node $(X_{10}, Y_{10})$ and $(X_{11}, Y_{11})$ each having two children of their own and so on. Now at each level, the size of the $X$-set gets halved. So, the number of levels in $T$ is at most $O(\log {n})$. However, the $Y$-sets need not get divided with equal proportions. So, at the frontier (the deepest level) we will have nodes of the above two types, depending on the size of their corresponding $Y$-sets. Let the collection of these frontier nodes be partitioned in two sets $\Phi$ and $\Psi$ corresponding to case 1 and case 2 respectively. 

We can conclude that the total number of probes needed to compute all relations in $\Phi$ is $o(n^2)$. This follows from Lemma 1. Here we can map the size of the $Y$-sets of the nodes in the collection $\Phi$ to the functions $f_i(n)$. We know that the total elements in the union of these $Y$-sets is $\le |Y_{00}| \le (8c+1/8c+2)n$. The total number of probes will be $F(n)$ in worst case. What is the total number of probes on the internal nodes of $T$? We know that in the internal nodes we compare the median of the $X$-set with the elements of the $A$-set, which takes $|A|$ probes. Since union of these $A$-sets cannot exceed the total number of vertices in $G$($n$), at each level of $T$ we do at most $O(n)$ probes, totaling to $O(n\log{n})$ probes over all the internal nodes.

 Unlike the nodes in $\Phi$, the nodes in $\Psi$ recursively calls the initial strategy using the input graph $G[Y]$. Let the probe complexity of our initial strategy be $Q(n)$. Then the recursion for $Q$ is as follows: \[Q(n) = \sum_{i=1}^{|\Psi|}{Q(n_i)} + o(n^2)\]
Here we assume that the nodes in $\Psi$ are indexed according to some arbitrary order. We can solve this recurrence using Lemma 2 giving $Q(n) \in o(n^2)$, since $\sum_{i=1}^{|\Psi|}{n_i} \le (8c+1/8c+2) n$. Note here that $|\Psi|$ is bounded by a constant since the size of  $Y$-sets are in $\Omega(n)$.

We call $\hat{T}$ the full tree. All leaf nodes in $\hat{T}$ are in $\Phi$. It is straightforward to show that $\hat{T}$ has  $O(\log^2{n})$ levels. Since any of the leaf nodes of $T$ has $|Y| \le \beta n$ (where $\beta = (8c+1/8c+2)$), its subtree in $\hat{T}$ can have at most $\alpha\log{\beta n} = \alpha\log{n} - \alpha\beta$ levels, and any of its leaves having at most $\alpha\log{n} - 2\alpha\beta$ levels and so on for some constant $\alpha$. 

\begin{figure}[h]
\includegraphics[width=11cm]{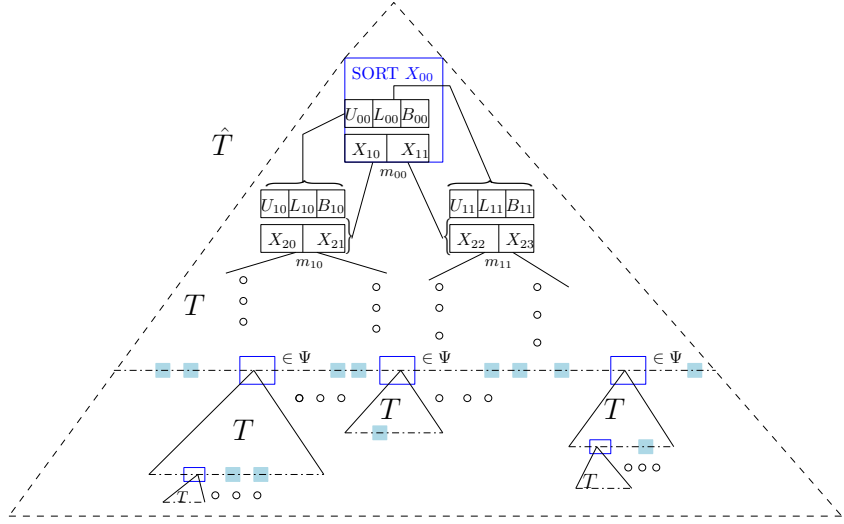}
\centering
\caption{Visualizing the steps. At the bottom of $T$ the shaded boxes represents the $\Phi$-nodes and the blue rectangles the $\Psi$-nodes. The outer dashed triangle represents the full tree $\hat{T}$. The tree $\hat{T}$ is created during the partitioning step and in the merge step we start from the deepest leaves of $\hat{T}$ and move upwards.} 
\end{figure}

\subsubsection{Merge step:}
Once we have completed building $\hat{T}$ we proceed with the final stage of our algorithm. Recall that during the forward partition step we had generated many of these $B$-sets in the internal nodes of $\hat{T}$. Now we start from the leaves of $\hat{T}$ and proceed upwards. Each pair of leaf nodes $l,r$ sharing a common parent $p$, sends a partial order to it them (computed as in case 1). When we merge this two partial orders we know that no extra probe is needed since they have already been split by the median of the $X$-set of $p$. What remains is to probe all edges between the $B$-set in $p$ and elements in this partial order (which constitutes the set of elements $A \cup X$ of the node $p$) as well as the edges in $G[B]$. Then we pass the resulting partial order to the parent of $p$, and so on. Since the size of the $B$-sets are bounded by $c_1$ (at any level in $\hat{T}$), total number of probes we make is then $\le c_1\sum_i(|A_i| + |X_i| + c_1)$. The sum is taken over all the nodes in that level. Hence this is bounded by $c_1n$, so at each level we do at most $O(n)$ probes in the backward merging step. Since there are at most $O(\log^2{n})$ levels, it totals to $O(n \log^2 n)$ additional probes. Adding this to the probe cost of partitioning in the forward step does not effect the total probe complexity, which was $o(n^2)$. The final step is to compute the transitive closure of the resulting set of relations, which can be done without any additional probing. Since computing the transitive closure is equivalent to boolean matrix multiplication\cite{23} the total complexity is $O(n^\omega)$. 
%------------------------------------------------------------------------------
%
%
%
%------------------------------------------------------------------------------
\subsection{The General Case}
We will define the sets $R$ and $S$ analogously to section 2.1. We have, $R = \{v \in V\ |\ n(v) > c_1q/n \}$ for some constant $c_1$. With $c_1 = 4$, we get $|R| \le \delta_1 n$ where $\delta_1 \le 2/c_1 = 1/2$. Hence $|S| \ge (1 - \delta_1)n \ge n/2$. Now we will apply Claim 1 successively to construct a ``big-enough'' set $X \subset S$ which we will use to find an approximate median of $V$. This set $X$ consists of disjoint subsets $X_i$ such that $G[X_i]$ is a clique. 

\subsubsection{Constructing $X$:} Let us define $S_i = S \setminus \bigcup_{j=1}^i{X_j}$. We construct the first clique $X_1 \subset S$ using the method detailed in Claim 1. There are two cases: 

\begin{enumerate}
   \item[Case 1] $q < n$: In this case we can show that $|X_1| \ge (n/2) / (c_1q/n + 1)  \ge n/10$. We take the first $n/10$ elements and keep the rest for the second round. Now we construct the second clique $X_2$ from $S_1$ which has at least $2n/25$ vertices. We let $X = X_1 \cup X_2$. Hence $X$ has at least $9n/50$ vertices.
   \item[Case 2] $q \ge n$:  In this case we have $|X_1| \ge (n/2) / (c_1q/n + 1) \ge n^2/10q$. Again we take $|X_1| = (1 / 10)n^{2}/q$ discarding some vertices if necessary. Similarly we construct $X_2 \subset S_1$. It can be shown that $|X_2| \ge (n^{2} / 10q)(1 - n/5q)$ and we keep $(n^{2} / 10q)(1 - n/5q)$ vertices in $X_2$  and the rest are discarded to be processed the next round. In general for the $r^{\mbox{th}}$ clique $X_r$ we have $|X_r| \ge (n^{2} / 10q)(1 - n/5q)^{r-1}$. Now we let $X = \bigcup_{i=1}^r{X_i}$. We will show that $|X| \ge \delta_2n$ for some constant $\delta_2 > 0$. We let $r = 5q/n + 1$. Then we have 

\begin{align}
  \nonumber  |X_r| &\ge (n^{2} / 10q)(1 - n/5q)^{r-1} \ge (n^{2} / 10q)(1 - n/5q)^{5q/n} > 3n^{2} / 100q
\end{align}

\noindent since $q \ge n$. Hence, $|X| = \sum_{i=1}^r{|X_i|} \ge r|X_r| \ge (9/50)n$, giving $\delta_2 = 9/50$. Now for each $X_i$ ($1 \le i \le r$) we keep a subset $Y_i$ of size $|X_r|$ and throw away the rest. Clearly, for each $i$, the induced sub-graph $G[Y_i]$ is also a clique. Let $Y = \bigcup_{i=1}^r{Y_i}$. We also have $|Y| \ge (9/50) n$. 

\end{enumerate}

 \subsubsection{Computing an approximate median of $V$:} We shall compute an approximate median with respect to all the vertices (the set $V$) and not just the set $S$. We will find a median element that divides the set $V$ in constant proportions. This can be done easily using the set $Y$. For each $Y_i$ we find its median using $\Theta(|Y_i|)$ probes since $G[Y_i]$ is a complete graph. Let this median be $m_i$ and $M$ be the set of these $r$ medians. Since $m_i \in S$, $n(m_i) \le 4q/n$. We define the upper set of $m \in M$ with respect to a set $A \subset V$ ($m$ may not be a member of $A$) as $U(m, A) = \{a \in A\ |\ a > m\}$. Similarly we define the lower set $L(m, A)$. We want to compute the sets $U(m, Y)$ and $L(m, Y)$. However, $m$ may not be neighbors of all the elements in $Y$. So we compute approximate upper and lower sets by probing all the edges  in $E(\{m\}, Y\setminus\{m\})$. These sets are denoted by $\widetilde{U}(m, Y)$ and $\widetilde{L}(m, Y)$ respectively. It is easy to see that there exists some $m \in M$ which divides $Y$ into sets of roughly equal sizes (their sizes are a constant factor of each other). In fact the median of $M$ is such an element. However the elements in $M$ may not all be neighbors of each other hence we will approximate $m$ using the ranks of the elements in $M$ with respect to the set $Y$ (which is $|\widetilde{L}(m, Y)|$). Next we prove that the element $m^*$ is an approximate median of $M$, picked using the above procedure, is also an approximate median of $Y$.
 
\begin{myclm}
    The element $m^*$ picked as described above is an approximate median of $Y$.
\end{myclm}
 
\begin{proof}
First we show that the median of $M$ is an approximate median of $Y$. This can be easily verified. Let us take the elements in $M$ in sorted order $(m_1,...,m_r)$, so the median of $M$ is $m_{\lfloor r/2 \rfloor}$. Now $L(m_{\lfloor r/2 \rfloor}, Y) = \bigcup_{i=1}^{\lfloor r/2 \rfloor}{L(m_i, Y_i)}$. Since, the sets ${Y_i}$ are disjoint and $L(m_i, Y_i) \ge |X_r|/2$, we have $|L(m_{\lfloor r/2 \rfloor}, Y)| \ge |X_r| r/4$ (ignoring the floor). Similarly we can show that $|U(m_{\lfloor r/2 \rfloor}, Y)| \ge |X_r| r/4$. Hence $m_{\lfloor r/2 \rfloor}$ is an approximate median of $Y$. Now we show that $\mid|L(m^*, Y)| - |L(m_{\lfloor r/2 \rfloor}, Y)|\mid < 4q/n$. Consider the sorted order of elements in $M$ according to $|\widetilde{L}(m^*, Y)|$. Since each element in $m \in M$ has at most $4q/n$ missing neighbors in $Y$, we have $\mid|\widetilde{L}(m, Y)| - |L(m, Y)|\mid < 4q/n$. So the rank of an element in the sorted order is at most $4q/n$ less than its actual rank. Thus an element $m^*$ picked as the median of $M$ using its approximate rank $|\widetilde{L}(m, Y)|$ cannot be more than $4q/n$ apart from $m_{\lfloor r/2 \rfloor}$ in the sorted order of $Y$. Hence,

\begin{align}
   |L(m^*, Y)| \ge |X_r| r/4 - 4q/n \ge 9n/200 - 4q/n \ge  n/40
\end{align}

\noindent whenever $n^2 \ge 200 q$. In an identical manner we can show that $|U(m^*, Y)| \ge n/40 $. Hence, $m^*$ is an approximate median of $Y$. When $q < n$ we just take $m^*$ as the median with the higher $|\widetilde{L}(\cdot,Y)|$ value, which guarantees $|L(m^*, Y)| \ge n/40$ whenever $n^2 \ge 800q/13$. So we take $n^2 \ge 200 q$ to cover both the cases.
\end{proof}
  
\noindent It immediately  follows that $m^*$ is also an approximate median of $V$ with both $|L(m^*, V)|$ and $|U(m^*, V)|$ lower bounded by $n/40$. Lastly, we note that the above process of computing an approximate median makes $\Theta(q + n)$ probes. This follows from the fact that computing the medians makes $\Theta(n)$ probes in total and for each of the $\le 5q/n + 1$ medians we make $O(n)$ probes.
 
\subsubsection{A divide-and-conquer approach:} Now that we have computed an approximate median of $V$ we proceed with an recursive approach.  Let $m^*$ be the median. As in section 3.1 we partition $V$ into three sets $U$, $L$ and $B$. The $U$ and $L$ are the upper and lower sets with respect to $m^*$. $B$ is the set of vertices that do not fall into either, that is, they are non-neighbors of $m^*$. Since $m^* \in S$ we have $|B| \le  4q/n$. We recursively proceed to partially sort the sets $U$ and $L$ with the corresponding graphs $G[U]$ and $G[L]$ and keep $B$ for later processing (as we did in the merging step previously). Like before we can imagine a recursion tree $T$. Let $E_{f_P}$ be the set forbidden edges in $G[P]$. We take $n_P = |P|$ and $q_P = |E_{f_P}|$.
 For each node $P \in T$ there are two cases:
\begin{itemize}
   \item[Case 1:] When $n_P^2 \ge 200q_P$, we recursively sort $P$. In this case we can guarantee that the approximate median $m_P^*$ of $P$ will satisfy equation (1). That is both $|L(m_P^*, P)|$ and $|U(m_P^*, P)|$ is $\ge n_P/40$.    
   \item[Case 2:] Otherwise we probe all edges in $G[P]$. In this case $P$ will become a leaf node in $T$.
\end{itemize}

\noindent It can be easily seen that the depth of the recursion tree is bounded by $O(\log{n})$ since at each internal node $P$ of $T$ we pass sets of constant proportions (where the size of the larger of the two set is upper bounded by $(39/40)n_P$) to its children nodes.

\subsubsection{Merge Step:} In this step we start with the leaves of $T$ and proceed upwards. A parent node $P$ gets two partial orders from its left and right children respectively. Then it probes all the edges between its $B$-set and these partial orders to generate a new partial order and pass it on to its own parent. This step works exactly as the ``merge step'' of the previous algorithm. Only difference is that the $B$-sets here may not be of constant size but of size $\le 4q/n$.

\subsubsection{Probe Complexity:} We can determine the probe complexity by looking at the recursion tree $T$. First we compute it for the forward partition step. At each internal node of $T$ we compute a set of medians and pick one element from it appropriately chosen. Then we partition the set of elements at the node by probing all edges between the selected element and rest of the elements in the node. As mentioned before this only takes $\Theta(q_P + n_P)$ probes for some internal node $P$. We assume that all the leaves of $T$ are at the same depth, otherwise we can insert internal dummy nodes and make it so. At each level of $T$ the sum total of all the vertices in every node is $\le n$ and the sum total of the forbidden edges is $\le q$. Hence we do $O(q + n)$ probes at any internal level of $T$. So for a total of $O(\log n)$ internal levels in $T$ the number of probes done is $((q + n)\log n)$ in the forward partition step. If $P$ is a leaf node then we probe all edges in $G[P]$. There are at most ${n_P \choose 2} - q_P$ edges in $G[P]$. Since $P$ is a leaf node, according equation 1, $n_P^2 < 200q_P$. Hence we make ${n_P \choose 2} - q_P = O(q_P)$ probes. Summing this over all the leaves gives a total of $O(q)$ probes. Hence the total probe complexity during the forward step is $O((q+n)\log n)$.

Now we look at the merging step. Merging happens only at the internal nodes. Lets look at an  arbitrary internal level of $T$. At each node $P$ of this level we probe all the edges in $E(B_P, U_P\cup L_P\cup m^*_P)$ and in $ G[B_P]$. Note that we do not have to make any probes between $U$ and $L$ as they were already separated by the approximate median $m^*_P$. Hence the total number of probes made in this node is $\le (|U_P| + |L_P| + |B_P| + 1)|B_P| \le (n_p)(4q_P/n_p) \le 4q_p$. Summing over all the nodes at any given level gives us $O(q)$ as the probe complexity per level. So the total probe complexity in the merging stage is $O(q\log{n})$. Hence, combining the probes made during the partition step and the merge step we see that the total probes needed to sort $V$ is $O((q+n)\log{n})$. 

\subsubsection{Total Complexity:} Now we look at the total complexity of the previous procedure. Again the analysis is divided into forward step and the merge step. In the forward step at each node $P$ we perform  $O(n_p^2)$ operations. This includes computing the degrees, finding the cliques, computing the approximate median. So at any level of $T$, regardless of it being an internal level or not, we perform $O(n^2)$ operations. Hence it totals to $O(n^2\log n)$ operations in the forward step. However this is a conservative estimate and we can remove the $\log n$ factor as argued below: we can define the recurrence for the forward computation as,
\begin{align}
  T(n)= \begin{cases} 
      T(n/40) + T(39n/40) + O(n^2) & n^2 \ge 200q \\
      O(q)\ \mbox{Otherwise}
   \end{cases}
\end{align} 

This follows from the previous discussion. If we don't recurse on a node we guarantee that $n_p^2 < 200 q_p$ for that node. Hence, we have $T(n) = O(n^2 + q)$ using the Akra-Bazzi method\cite{22}. In the merge step,  we only make $O(q_p)$ comparisons at any given node. We compute transitive closures only at the leaves. However for  any leaf $P$ we have $n_P^2 < 200q_P$. Hence computing the transitive closure of $G[P]$ takes $O(q_P^{\omega/2})$ time. Hence, the total complexity of the above procedure is $O(n^2 + q^{\omega/2})$. We summarize the results in this section with the following theorem:

\begin{mythm}
   Given a graph $G(V,E)$ of $n$ vertices having $q$ forbidden edges, one can compute the partial order of $V$ with $O((q+n)\log {n})$ comparisons and in total $O(n^2 + q^{\omega/2})$ time.
\end{mythm}
\begin{proof}
   Follows from the discussions in this section.
\end{proof}
%------------------------------------------------------------------------------
%
%
%
%------------------------------------------------------------------------------
 
\section{A Randomized Algorithm}
In this section we look at a more direct way of sorting by making random probes. The proposed method is inspired by the literature on two-step oblivious parallel sorting \cite{10,11} algorithms, in particular on a series of studies by Bollob\'{a}s and Brightwell showing certain sparse graphs can be used to construct efficient sorting networks \cite{3,12}. It was shown that if a graph satisfies certain properties then probing its edges and taking the transitive closure of the resulting set would yield large number of relations. Then we just probe the remaining edges that are not oriented, which is guaranteed (with high probability) to be a ``small'' set.

The main idea is as follows: Let $\mathcal{H}_n$ be a collection of undirected graphs on $n$ vertices having certain properties. A transitive orientation of a graph $H(V, E) \in \mathcal{H}_n$ is an ordering of $V$ and the induced orientation of the edges of $H$ based on that ordering. Let $\sigma$ be an ordering on $V$ and $P(H, \sigma)$ be the partial order generated by this ordering $\sigma$ on $H$. It is a partial order since $H$ may not be sortable. Let $\mathcal{P} = P(H, \sigma)$ and $t(P)$ be the number of incomparable pairs in $\mathcal{P}$. We want $H$ to be such that $t(p)$ is small. If that is the case then $\mathcal{P}$ will have many relations and if $H$ is sparse then we can probe all the edges of $H$ and afterwards we will be left with probing only a small number of pairs. These are pairs which were not oriented during the first round of probing and after the transitive closure computation. A graph $H$ is \textit{useful} to our purpose if every transitive orientation of $H$ results in many relations. We want to find a collection $\mathcal{H}_n$ such that every graph in it is useful with high probability. 

We extend the results in \cite{3,12} to show that a collection of certain conditional random graphs are useful, with high probability. In our case this random graph will be a spanning subgraph of the input graph $G$. Here we recall an important result  from \cite{3} which we will use in our  proof.

\begin{mythm}[Theorem 7 in \cite{3}]
   If  $G$ is any graph on $n$ vertices and $G$  satisfies the following property:
   \begin{enumerate}
      \item[Q1:] Any two subsets $A, B$ of vertices having size $l$ have at least one edge between them.
   \end{enumerate}
   Then, the number of incomparable pairs in $P(G, \sigma)$ is at most $O(nl\log{l})$ for any $\sigma$.
\end{mythm}

\noindent The input graph $G$ is chosen by our adversary. However, we show that any random spanning subgraph of $G$ with an appropriate  edge probability will satisfy Q1 with high probability. Let $H_{n,p}(G)$ be a random spanning subgraph of $G$, where $H_{n,p}(G)$ has the same vertex set as $G$ and a pair of vertices in $H_{n,p}(G)$ has an edge between them with probability $p$ if they are adjacent in $G$, otherwise they are also non-adjacent in $H_{n,p}(G)$. All we need to prove is that any random spanning subgraph $H_{n,p}(G)$ given $G$ with $n$-vertices and edge probability $p$ will satisfy Q1 with high probability. Since $G$ has at most $q$ forbidden edges any two subsets of vertices $A, B$ (not necessarily distinct) of size $l$ must have at least ${l \choose 2} - q$ edge between them. Let,  $E_{AB}$  be the event that the pair $(A,B)$ is bad (they have no edges between them), then the probability $S_{n,p}$ that there exists a bad pair is:

\begin{align}
   S_{n,p} := \mathbb{P}(\sum_{i,j}{E_{A_iB_j}}) \le \sum_{i,j}\mathbb{P}{(E_{A_iB_j})}  \le \sum_{i,j}{(1 - p)^{e(A_i,B_j)}}
\end{align}

\noindent where the sum is taken over all such ${n \choose l}^2$ pairs of subsets, and the number of edges between the two sets $A$ and $B$ in $G$ is $e(A, B) \ge {l \choose 2} - q$. So we have,
\begin{align}
   \nonumber S_{n,p} &\le {n \choose l}^2(1 - p)^{{l \choose 2} - q} \le {n \choose l}^2e^{-p({l \choose 2} - q)}\ \ \mbox{Since, $e^{-x} \ge 1- x$}\\
   \nonumber &\le  \left({{en} \over l}\right)^{2l}e^{-p({l \choose 2} - q)} \le \exp(2l(\log{en/l}) - p({l \choose 2} - q))
\end{align} 

\noindent Hence $S_{n,p} \to 0$ as $n \to \infty$ whenever $\exp(2l(\log{en/l}) - p({l \choose 2} - q)) = o(1)$. Given $q < {n \choose 2}$ it is always possible to find appropriate values for $p$ and $l$ as functions of $q$ and $ n$ such that $S_{n,p} = o(1)$. Given some value for the pair $(p, l)$, we see that in the first round we make $O(pn^2)$ probes with high probability and in the second round $O(nl\log{l})$ probes again with high probability. So the total probe complexity is $\widetilde{O}(pn^2 + nl)$. With some further algebra it can be shown that this is $\widetilde{O}(n^2/\sqrt{q + n} + n\sqrt{q})$. We summarize this section with the following theorem:

\begin{mythm}
   Given a graph $G$ on $n$ vertices and $q$ forbidden edges one can determine the partial order on $G$ with high probability in two steps by probing only $\widetilde{O}(n^2/\sqrt{q + n} + n\sqrt{q})$ edges in total and in $O(n^\omega)$ time.
   \end{mythm}
\begin{proof}
   Follows from the preceding discussions.
\end{proof}

\subsection{When $G$ is a Random Graph}
The above technique can easily be extended for the case when the input graph is random. Let $G_{n,p}$ be the input graph having $n$-vertices and an uniform edge probability $p$. For such a graph we can use equation (3) to bound $S_{n,p}$ as follows:
 
\begin{align}
   \nonumber S_{n,p} &\le {n \choose l}^2(1 - p)^{l^2} \le \exp(-pl^2 + 2l\log{n}) 
\end{align} 

Hence, we can choose any $l > 2\log{n}/p$ such that $S_{n,p} \to 0$ as $n \to \infty$. Let $l = 3\log{n}/p$. Using Theorem 2 we have $t(G_{n,p}) = \widetilde{O}(nl) = \widetilde{O}(n/p)$. Since $G_{n,p}$ has $pn^2/2$ edges (with high probability) the critical value of $p$ when $t(G_{n,p}) = pn^2/2$ is $\widetilde{O}(1/\sqrt{n})$. Let this be $\hat{p}$. Hence if $p > \hat{p}$, we can sort by making only $\widetilde{O}(n^{3/2})$ comparisons. Since given $G_{n,p}$ we can construct an induced subgraph $G_{n,\hat{p}}$ and use it as the random graph in our previous construction. Otherwise we just probe all the edges which makes $O(pn^2)$ comparisons. Thus we can sort $G_{n,p}$ with at most $\widetilde{O}(\min{(n^{3/2}, pn^2)})$ comparisons with high probability. Hence, we get an elementary technique to sort a random graph with at most $\widetilde{O}(n^{3/2})$ comparisons. The algorithm in \cite{1} has a slightly better bound of $\widetilde{O}(n^{7 / 5})$ comparisons. However, the total runtime of the algorithm in \cite{1} is only polynomially bounded when $p$ is small. In our algorithm we need compute the transitive closure only twice making it run in $O(n^\omega)$ total time.

%------------------------------------------------------------------------------
%
%
%
%------------------------------------------------------------------------------
\section*{Concluding Remarks}
In this paper we study the problem of sorting under non-uniform comparison costs, where costs are either 1 or $\infty$. This cost structure is non-monotone resulting in additional complexity. The results presented here only uses elementary techniques, yet achieving non-trivial bounds on probe complexity. Further, we present strong evidence that the complexity of sorting $V$ is dependent on certain properties of the input graph, in particular the number of forbidden edges $q$. We derive an non-trivial upper bound $O((q + n)\log {n})$ for the probe complexity. The total complexity of our algorithm is bounded by $O(n^2 + q^{\omega/2})$. Since the lower bound for the total complexity of the problem is $\Omega(n^2)$, module fast matrix multiplication, the proposed algorithm is almost optimal in terms of the total complexity. We also present a randomized algorithm for the problem which uses $\widetilde{O}(n^2/\sqrt{q + n} + n\sqrt{q})$ probes with high probability. When the input graph is random this algorithm requires only $\widetilde{O}(n^{3/2})$ probes again with high probability. 
%------------------------------------------------------------------------------
%
%
%
%------------------------------------------------------------------------------

%------------------------------------------------------------------------------
%
%
%
%------------------------------------------------------------------------------

\section*{Appendix}
\subsubsection{Proof of Lemma 1:\\}

Let, $\{f_1(n),f_2(n),...,f_k(n) \}$ be a finite set of non-negative monotonically increasing functions in $ n $ such that:
	\begin{enumerate}
		\item $ \forall i\ f_i(n) \in o(g(n)) $ 
		\item $ \sum_{i}{f_i(n)} \le cg(n)$
	\end{enumerate}
	
If $ F(n) = \sum_{i}{f_i^2(n)}$	Then $F(n) \in o(g^2(n))$.
\begin{proof}
First we prove $F(n) = O(g(n))$. Clearly,
\begin{align}
   \nonumber (\sum_{i}{f_i(n)})^2 \le c^2g^2(n) \\
   \nonumber \sum_{i}{f_i^2(n)} + 2\sum_{i,j}f_i(n)f_j(n) \le c^2g^2(n) \\
   \nonumber F(n) \le c^2g^2(n)
\end{align}

  Now we prove $F(n) \not = \Omega(g^2(n))$: assume that $F(n) \in \Omega(g^2(n))$, then there exists some constant $\hat{c}$ such that, $F(n) \ge \hat{c}g^2(n)$ whenever $n \ge n_1$. Now let $f_i(n) \le c_ig(n)$ whenever $n \ge n_0(c_i)$. Since, $f_i(n) \in o(g(n))$ we can pick this $c_i$'s arbitrarily and independent of each other. Now, for $n \ge \max(n_1, n_2)$ (where $n_2 = \max_i(n_0(c_i))$) we have,
  
  \begin{align}
    \nonumber  \sum_{i}{f_i^2(n)} \ge \hat{c}g^2(n)\\
    \nonumber \sum_{i}{c_i^2} \ge \hat{c}
  \end{align}
This contradicts the fact that $c_i$'s can be assigned arbitrary values independent of each other. That is, not all $f_i(n)$ will satisfy the condition $f_i(n) \in o(g(n))$ simultaneously. Hence,  $F(n) \not = \Omega(g^2(n))$.

\end{proof}

\subsubsection{Proof of Lemma 2:\\}

Let, $T(n) = \sum_{i=1}^{k}{T(n_i)} + f(n)$ where $\sum_i{n_i} \le \delta n$ for some $0 < \delta < 1$ and $f(n) \in o(n^2)$. Then, $T(n) \in o(n^2)$.

\begin{proof}
Let as assume $T(n) = \Omega(n^{\alpha})$ for some $\alpha \ge 1$. Otherwise we are done. Hence, $\sum_{i=1}^{k}{T(n_i)} \le T(\sum_{i=1}^{k}{n_i}) = T(\delta n)$. So, the recurrence becomes,  $T(n) \le T(\delta n) + f(n)$. Using Master theorem we see that the case 3 applies here, which gives, $T(n) = \Theta(f(n)) = o(n^2)$.
\end{proof}
\end{document}